\documentclass[pdflatex,sn-mathphys-num]{sn-jnl}
\usepackage{amsfonts, amsmath, amssymb, amsthm, tikz, float, url,  mathrsfs,array,graphicx,pgfplots,bm, enumitem}

\usepackage[capitalise,noabbrev]{cleveref}
\usepackage[final]{microtype}
\usepackage[dvipsnames]{xcolor}
\usepackage[lined,boxed,commentsnumbered,ruled,linesnumbered]{algorithm2e}
\usepackage[all]{hypcap}
\usepackage[none]{hyphenat}

\pgfplotsset{compat=1.18}
\usetikzlibrary{patterns,decorations.pathreplacing,decorations.markings,arrows, arrows.meta}
\Crefname{algocf}{Algorithm}{Algorithms}
\providecolor{DarkBlue}{rgb}{0,0,.545}
\providecolor{DarkGreen}{rgb}{0,.392,0}
\hypersetup{citecolor=DarkGreen}
\hypersetup{linkcolor=DarkBlue}
\hypersetup{urlcolor=DarkBlue}

\newtheorem{theorem}{Theorem}
\newtheorem{example}{Example}
\newtheorem{definition}{Definition}
\newtheorem{lemma}{Lemma}

\AtBeginDocument{%
	\def\doi#1{\url{https://doi.org/#1}}}

\begin{document}

\title{Shortest Paths in a Weighted Simplicial Complex}

\author{\fnm{Sukrit} \sur{Chakraborty\href{https://orcid.org/0000-0003-4678-1829}{\includegraphics[height=\fontcharht\font`B]{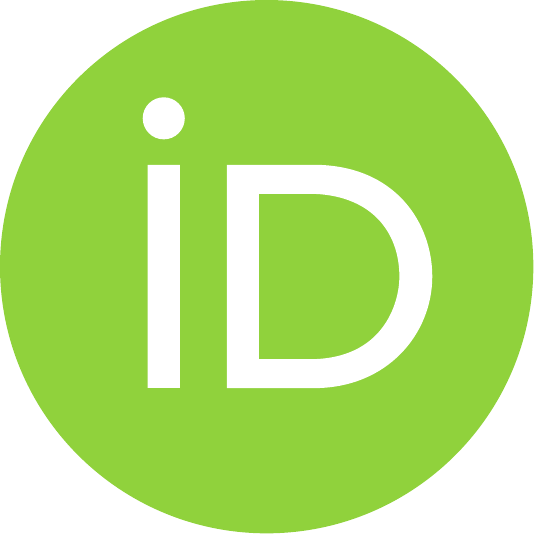}}}}\email{sukrit049@gmail.com}

\author{\fnm{Prasanta} \sur{Choudhury}}\email{prasantachoudhury98@gmail.com}

\author*{\fnm{Arindam} \sur{Mukherjee*\href{https://orcid.org/0000-0001-5505-6536}{\includegraphics[height=\fontcharht\font`B]{orcid.pdf}}}}\email{arindam@smail.iitm.ac.in}

\affil{\orgname{Achhruram Memorial College}, \orgaddress{\city{Jhalda}, \state{West Bengal}, \country{India}, \postcode{723202}}}

\abstract{Simplicial complexes are extensively studied in the field of algebraic topology. They have gained attention in recent time due to their applications in fields like theoretical distributed computing and simplicial neural networks. Graphs are mono-dimensional simplicial complex. Graph theory has application in topics like theoretical computer science, operations research, bioinformatics and social sciences. This makes it natural to try to adapt graph-theoretic results for simplicial complexes, which can model more intricate and detailed structures appearing in real-world systems. Though seemingly obvious, we did not find any previous work that looked into this prospect of simplicial complexes.

In this article, we define the concept of weighted simplicial complex and $d$-path in a simplicial complex. Both these concepts have the potential to have numerous real-life applications. We start by adapting the Depth-First Search and Breadth-First Search algorithms for our setup. Next, we provide two novel algorithms to find the shortest paths in a weighted simplicial complex. The core principles of our algorithms align with those of Dijkstra’s algorithm and Bellman-Ford algorithm for graphs. Hence, this work lays a building block for the sake of integrating graph-theoretic concepts with abstract simplicial complexes.}


\maketitle

\section{Introduction}
Simplicial Complexes are objects that are studied extensively in algebraic topology using tools like homology and cohomology. Simplicial complexes can capture intricate topological properties, such as holes and voids, across multiple dimensions. Simplicial complexes are widely applied in algebraic topology, computational geometry, and data analysis, particularly for capturing the shape of data (e.g., in persistent homology) and for the study of spaces with more complex topological structures. On the other hand, an abstract simplicial complex is the combinatorial counterpart of a simplicial complex. Abstract simplicial complex is often studied algebraically, thus forming a relation between combinatorics and commutative algebra. We formally define abstract simplicial complexes in the next section. 

Now a graph can be seen as a boiled down version of a simplicial complex. While a graph is typically a 1-dimensional object, an abstract simplicial complex can exist in any finite dimension. One of the main differences between a graph and a simplicial complex is that in a simplicial complex with $d \geq 2$, more than two \( (d-1) \)-dimensional simplices can be connected by a single \( d \)-dimensional simplex, whereas in a graph, an edge can connect at most two vertices. 

Graphs are commonly used in computer science, network theory, social science, and many other fields to represent relationships between entities, such as in social networks, computer networks, and transportation systems. A graph is usually analyzed in terms of its connected components, paths, and cycles.


On the other hand, Simplicial complexes have also been used in various fields as detailed in \cite{Herlihy99,Herlihy14,Koz12}. For example, in \cite{Herlihy99} and \cite{Herlihy14},  the results on simplicial complexes were used to prove new results in theoretical distributed computing. Distributed computing is a field of computer science that studies distributed systems. In \cite{Koz12}, Kozlov studied the chromatic subdivision of a simplicial complex. Chromatic subdivision plays a central role in the analysis of the solvability of tasks in different computational models in computer science (see \cite{attiya2004distributed}  for details). Simplicial neural networks (SNN) have emerged as a new direction in graph learning, extending the concept of convolutional architectures from node space to simplicial complexes, as detailed in \cite{ebli2020simplicial, chen2022bscnets, wu2023simplicial}. The reconstruction of random simplices (see \cite{adhikari2022shotgun, ophelders2025sweeping, ma2025reconstructing}) has emerged as a recent and significant trend in the study of random structures, gaining attention for its theoretical importance and potential applications.

Hence, to better equip abstract simplicial complexes to counter more real life problems, extending results from graphs to abstract simplicial complexes is a natural progression, as the latter provide a more versatile and detailed framework for modeling the complex structures observed in many real-world phenomena. Consequently, deeper insights into these complex structures and their underlying properties can be gained by extending graph-theoretic methods to simplicial complexes.

\subsection{An Example of Real-life Application of Simplicial Complexes}\label{sec:example}
  Consider a town with housing complexes \( A \), \( B \), and \( C \), where each complex contains two housing flats: \( A1, A2, B1, B2, C1, C2 \). A delivery boy is tasked with delivering parcels to each flat. The delivery boy can travel from the delivery office (denoted \( O \)) to a specific complex and deliver parcels to multiple flats within that complex.

In this setup:

\begin{itemize}
    \item Each flat (\( A1, A2, B1, B2, C1, C2 \)) and the delivery office (\( O \)) is represented as a vertex.
    \item A single edge from \( O \) to a complex \( A \) represents the delivery to both \( A1 \) and \( A2 \).
\end{itemize}

This situation is better modelled using a simplicial complex rather than a graph, as a simplicial complex allows for the representation of higher-dimensional relationships. Specifically, a simplicial complex can represent subsets of vertices (e.g., \( \{O, A1, A2\} \)) as simplices, naturally capturing the interaction of multiple entities (flats) within the same complex.

\subsection{Our Contributions}

The concept of assigning weights to edges in a graph has immense practical importance, forming the basis of many optimization problems in real-world systems. Motivated by this, in this article we introduce two fundamental generalizations to the higher-dimensional setting: (i) the notion of a \emph{weighted simplicial complex}, and (ii) the definition of a \emph{$d$-path} in such a complex. Both concepts provide the groundwork for extending classical graph-theoretic algorithms to simplicial complexes, enabling the treatment of higher-dimensional adjacency relations in a combinatorial framework. To the best of our knowledge and much to our astonishment, these concepts have not been defined in the literature before.

Building upon these foundations, our main contributions can be summarized as follows:

\begin{enumerate}
    \item \textbf{Generalized Discovery Algorithms:} We extend the classical \emph{Breadth-First Search (BFS)} and \emph{Depth-First Search (DFS)} algorithms from graphs to simplicial complexes. These algorithms, presented as \cref{alg:bfs-d} and \cref{alg:dfs-d}, can be thought of as the warm-up for the logical thought chain to our later more intricate algorithms. We also prove their correctness and analyze their time and memory complexities, both of which remain linear in the size of the complex-matching the optimal asymptotic bounds known for graphs.
    
    \item \textbf{Shortest-Path Algorithm in Weighted $d$-Complexes:} We propose a novel greedy algorithm (see \cref{alg:dijk-simplicial-complex}) for computing the shortest $d$-paths in a weighted simplicial complex. Our algorithm generalizes \emph{Dijkstra’s algorithm} to the simplicial setting while retaining its essential principles of distance relaxation and greedy vertex selection. We rigorously prove its correctness and derive its time and memory complexities, showing that the algorithm runs in $\mathcal{O}((n + m)\log n)$, where $n$ and $m$ denote the numbers of $(d-1)$- and $d$-simplices respectively.

    \item \textbf{Relaxation-Based Generalization of Bellman-Ford:} Extending beyond non-negative weights, we introduce \cref{alg:bf-d}, a Bellman-Ford-type procedure that correctly computes shortest weighted $(d-1)$-paths even in the presence of negative weights. This algorithm also detects \emph{negative $d$-cycles}, thus fully generalizing the Bellman-Ford framework to higher-dimensional combinatorial structures.

\end{enumerate}

This work bridges the gap between classical graph algorithms and higher-dimensional combinatorial structures. This offers theoretical depth and potential for real-world applications. Our constructions illustrate how simplicial complexes can serve not only as topological objects but also as computational and algorithmic entities in discrete mathematics and network science. We also give an example where our set-up can model the problem in hand in a better way than usual graphs.

\subsection{Outline}
In \cref{sec:prelim}, we provide all the notations and definitions necessary for explaining our algorithm. This section also introduces the concepts of weighted simplicial complex and $d$-path. The \cref{sec:graph-discovery} includes the description and analysis of BFS and DFS for our setting. In \cref{sec:algo} and \cref{sec:bellman} we describe our novel algorithms to find the shortest path in a simplicial complex. We provide a real-life example of our setup in \cref{sec:realexample} that shows how our setup can model some of the real-life problems better than graphs. We conclude our study with some remarks and possible future research directions in \cref{sec:conclu}.

\section{Preliminaries}\label{sec:prelim}
In this section, the notations and definitions that will form the building blocks for our algorithm will be discussed. This section also covers all the newly introduced concepts in this paper.

\subsection{The Set-up}
We start by formally defining a finite abstract simplicial complex.

\begin{definition}[Finite Abstract Simplicial Complex]
    Let \( \mathcal{X}_0 \) be a finite set. A finite abstract simplicial complex \( X \) on \( \mathcal{X}_0 \) is a collection of nonempty subsets of \( \mathcal{X}_0 \) such that:
\begin{enumerate}
    \item For every element \( x \in \mathcal{X}_0 \), the singleton \( \{x\} \) belongs to \( X \),
    \item If \( \sigma \in X \) and \( \tau \subset \sigma \), then \( \tau \in X \).
\end{enumerate}
\end{definition}

For example, \( X = \{\{1\}, \{2\}, \{3\}, \{1,2\}, \{2,3\}, \{1,3\}, \{1,2,3\}\} \) is an abstract simplicial complex on \(\mathcal{X}_{0} =  \{1,2,3\} \) (see \cref{fig:finite-abstract-simplicial-complex-example}). 
\begin{figure}[ht]
    \centering
    \scalebox{1}{\begin{tikzpicture}[scale = 1.2]

\tikzset{
    bluenode/.style={
        draw=blue,circle, inner sep =1pt
        },
    rednode/.style={
        },
}
  \node[bluenode] (a1) at (-2,0) {1};  
  \node[bluenode] (a2) at (2,0)  {2}; 
  \node[bluenode] (a3) at (0,3.46)  {3};  

  \draw (a1) -- (a2);  
  \draw (a1) -- (a3);  
  \draw (a2) -- (a3);

\filldraw[OliveGreen,opacity=.1] (-2,0) -- (2,0) -- (0,3.46) -- cycle;

\draw[->] (0,1.5) to[out=40,in=190] (2,2.5); 

\node at (2.5,2.5) {$\{1,2,3\}$};
\end{tikzpicture}  }
    \caption{Example of a finite abstract simplicial complex}
    \label{fig:finite-abstract-simplicial-complex-example}
\end{figure}
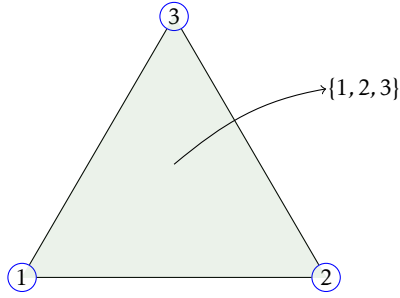

For any finite set \( A \), let \( |A| \) denote the number of elements in \( A \). A set \( \sigma \in X \) with \( |\sigma| = k+1 \) is called a \( k \)-dimensional simplex (or simply a \( k \)-simplex), where $k$ is a non-negative integer. Specifically, a vertex is a \( 0 \)-simplex, an edge is a \( 1 \)-simplex, a triangle is a \( 2 \)-simplex, and so on.

For brevity, we use the term \emph{complex} (or \emph{$d$-complex}) instead of abstract simplicial complex throughout this article. The dimension of a complex \( X \), denoted by \( \dim(X) \), is the maximum dimension of all simplices in \( X \), defined as
\[
\dim(X) := \max \{\dim(\sigma) : \sigma \in X\},
\]
where \( \dim(\sigma) = |\sigma| - 1 \). Note that if \( \dim(X) = 1 \), \( X \) can be considered as a graph whereas  \( \dim(X) = 0 \) represents a null graph.

For \( 0 \le j \le \dim(X) \), the set of all \( j \)-simplices of \( X \) is denoted by
\[
X_j := \{\sigma \in X \mid \dim(\sigma) = j\}.
\]
For the remainder of the article, for \( d \in \mathbb{N} \), the complex will take the form
\begin{equation}\label{eqn:complex}
	X := \left( \bigcup_{k=0}^{(d-1)} X_k \right) \cup X^d,
\end{equation}
where \( \mathcal{X}_0 := \{1, 2, \dots, n\} \), \( X_k := \{\{i_0, \dots, i_k\} \mid 1 \leq i_0 < \dots < i_k \leq n\} \) for \( 1 \le k \le d \), and \( X^d \subseteq X_d \). Here, \( X_k \) represents the set of all \( k \)-simplices on \( \mathcal{X}_0 \). In this model, \( X \) is a complex with a complete \( ((d-1)) \)-dimensional skeleton.

\begin{definition}[Weighted $d$-complex]
    A weighted $d$-complex \( (X,w) \) is a $d$-complex in which each member of $X^d$ is assigned a weight (or cost), typically represented as a real number. The weight of a $d$-simplex \( \tau \in X^d \) is denoted by \( w(\tau) \), where \( w: X^d \to \mathbb{R} \) is a weight function.
\end{definition}

We say that \( \sigma, \sigma' \in X_{(d-1)} \) are \emph{neighbors} if \( \sigma \cup \sigma' \in X^{d} \), in which case we write \( \sigma \sim \sigma' \) (see \cref{fig:weighted-2-complex-example}\footnote{Given a $d$-dimensional simplicial complex $X$, the figures illustrating $X$ display a selection of $(d-1)$-dimensional faces, rather than the exhaustive set of all $(d-1)$-simplices contained within $X$, in order to maintain visual concision and avoid clutter.}). This concept is similar to the one introduced in \cite{PRR2017}. 

\begin{figure}[ht]
    \centering
    \scalebox{.8}{\begin{tikzpicture}[scale = 1.2]

\tikzset{
    bluenode/.style={
        draw=blue,circle, inner sep =1pt
        },
    rednode/.style={
        },
}
  \node[bluenode] (a1) at (0,0) {1};  
  \node[bluenode] (a2) at (2,1)  {2}; 
  \node[bluenode] (a3) at (-2,1)  {3};  
  \node[bluenode] (a4) at (0,-1)  {4};
  \draw (a1) -- (a2);  
  \draw (a1) -- (a3);  
  \draw (a2) -- (a3);  
  \draw (a4) -- (a3);
  \draw (a2) -- (a4);
  \draw (a1) -- (a4);
  
\filldraw[Sepia,opacity=.1] (0,0) -- (2,1) -- (0,-1) -- cycle;

\filldraw[Periwinkle,opacity=.1] (0,0) -- (-2,1) -- (0,-1) -- cycle;

\draw[->] (.5,0) to[out=30,in=185] (3,.5); 

\node at (4,.5) {$\omega(\{1,2,3\})=5$};

\node at (4,0) { $X_{0} = \{\{1\},\{2\},\{3\},\{4\}\}$};
\node at (4,-.5) { $X_{1} = \{\{1,2\},\{1,3\},\{1,4\},\{2,3\},\{2,4\},\{3,4\}\}$};
\node at (4,-1) { $X^{2} = \{\{1,2,4\},\{1,3,4\}\}$};
\node at (4,-1.5) {Here $\{1,4\} \sim \{2,4\}$ since $\{1,2,4\}\in X^{2}$};

\draw[->] (-.5,0) to[out=150,in=-5] (-3,.5); 

\node at (-4,.5) {$\omega(\{1,3,4\})=3$};
\end{tikzpicture}  }
    \caption{Weighted 2-Complex and Neighbour}
    \label{fig:weighted-2-complex-example}
\end{figure}
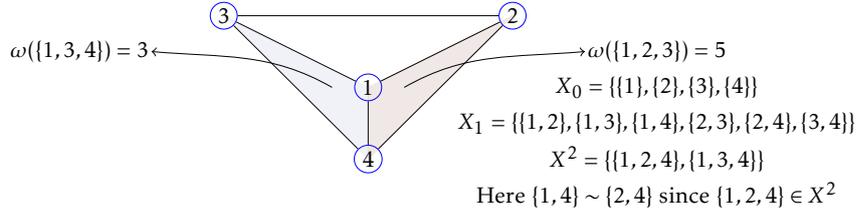

The \emph{degree} of a simplex \( \sigma \in X_{(d-1)} \) is denoted by
\[
\deg(\sigma) = \deg_X(\sigma) := \sum_{\tau \in X^d} \mathbf{1}_{\{\sigma \subset \tau\}},
\]
which counts the number of \( d \)-simplices containing \( \sigma \). Note that a simplex \( \tau \in X^d \) contributes to the sum if \( \tau = \sigma \cup \{v\} \) for some \( v \in \mathcal{X}_0 \setminus \sigma \). The set of neighbors of \( \sigma \in X_{(d-1)} \) is denoted by \( S_\sigma \), that is,
\[
S_\sigma = \{\sigma' \in X_{(d-1)} \mid \sigma' \sim \sigma\}.
\]

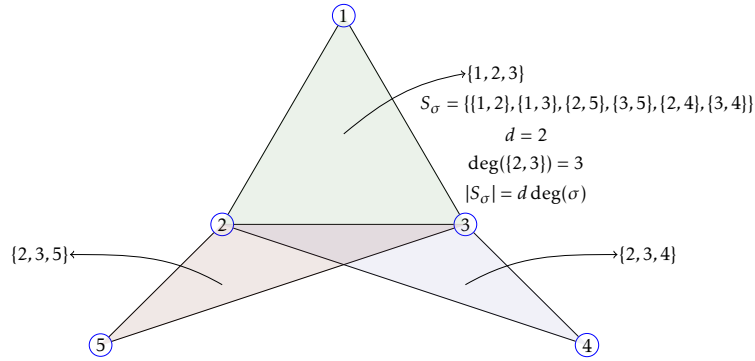
\begin{figure}[ht]
    \centering
    \scalebox{.75}{\begin{tikzpicture}[scale = 1.2]

\tikzset{
    bluenode/.style={
        draw=blue,circle, inner sep =1pt
        },
    rednode/.style={
        },
}
  \node[bluenode] (a2) at (-2,0) {2};  
  \node[bluenode] (a3) at (2,0)  {3}; 
  \node[bluenode] (a1) at (0,3.46)  {1};
  \node[bluenode] (a4) at (4,-2)  {4};
  \node[bluenode] (a5) at (-4,-2)  {5};

  \draw (a1) -- (a2);  
  \draw (a1) -- (a3);  
  \draw (a2) -- (a3);  
  \draw (a4) -- (a3);
  \draw (a2) -- (a4);
  \draw (a5) -- (a3);
  \draw (a2) -- (a5);

\filldraw[OliveGreen,opacity=.1] (-2,0) -- (2,0) -- (0,3.46) -- cycle;
\filldraw[Periwinkle,opacity=.1] (-2,0) -- (2,0) -- (4,-2) -- cycle;
\filldraw[Sepia,opacity=.1] (-2,0) -- (2,0) -- (-4,-2) -- cycle;

\draw[->] (0,1.5) to[out=40,in=190] (2,2.5); 
\draw[->] (2,-1) to[out=30,in=180] (4.5,-.5);
\draw[->] (-2,-1) to[out=150,in=0] (-4.5,-.5);

\node at (2.5,2.5) {$\{1,2,3\}$};
\node at (5,-.5) {$\{2,3,4\}$};
\node at (-5,-.5) {$\{2,3,5\}$};

\node at (3,1.5) {$d = 2$};
\node at (3,1) {$\deg(\{2,3\}) = 3$};
\node at (4,2) {$S_{\sigma} = \{\{1,2\},\{1,3\},\{2,5\},\{3,5\},\{2,4\},\{3,4\}\}$};
\node at (3,.5) {$|S_\sigma| = d \deg(\sigma)$};

\end{tikzpicture}  }
    \caption{Degree and Set of Neighbours}
    \label{fig:degree-set-of-neighbours-example}
\end{figure}

It follows that the number of elements in \( S_\sigma \) is \( d \) times \( \deg(\sigma) \), i.e.,
\[
|S_\sigma| = d \deg(\sigma).
\]

\begin{definition}[$d$-path and $d$-cycle]\label{d-path}
    A $d$-path in a $d$-complex $X$ is a sequence of distinct $(d-1)$-simplices  \( \sigma_0, \sigma_1, \dots, \sigma_l \) such that \( (\sigma_i \cup \sigma_{i+1}) \in X^d \) for all \( 0 \leq i < l \) and the collection $\{\sigma_i \cup \sigma_{i+1}: 0\leqslant i \leqslant l\}$ is a disjoint collection. The $d$-path is said to have length \( l \) and is denoted by \( P = (\sigma_0, \sigma_1, \dots, \sigma_l) \). If the $d$-path starts and ends at the same $(d-1)$-simplex, it is called a $d$-cycle of length $l$.
\end{definition}

\begin{figure}[ht]
    \centering
    \scalebox{1}{\begin{tikzpicture}

\tikzset{
    bluenode/.style={draw=blue,circle, inner sep=1pt},
}

  \node[bluenode] (a1) at (-1,0) {1};  
  \node[bluenode] (a2) at (0,0)  {2}; 
  \node[bluenode] (a3) at (1,0)  {3};  
  \node[bluenode] (a4) at (1.43,-0.86) {4};  
  \node[bluenode] (a5) at (.57,-.86)  {5};  
  \node[bluenode] (a6) at (-.43,-.86)  {6};  
  \node[bluenode] (a7) at (-1.43,-.86)  {7};  
  \node[bluenode] (a8) at (-1.86,-1.7) {8};  
  \node[bluenode] (a9) at (-1,-1.7)  {9};  
  \node[bluenode] (a10) at (0,-1.7)  {10};  
  \node[bluenode] (a11) at (1,-1.7)  {11};
  \node[bluenode] (a12) at (1.86,-1.7)  {12};
  \node[bluenode] (a13) at (1,-3)  {13};
  \node[bluenode] (a14) at (2.2,-2.5)  {14};
  
  \draw (a1)--(a2)--(a3)--(a4)--(a11)--(a5)--(a2);
  \draw (a1)--(a6)--(a2);
  \draw (a5)--(a6)--(a10)--(a11)--(a13);
  \draw (a1)--(a7)--(a6);
  \draw (a7)--(a8)--(a9)--(a6);
  \draw (a9)--(a10);
  \draw (a11)--(a12)--(a14)--(a13);
  \draw (a11)--(a14);
  \draw (a6)--(a8);

  \filldraw[OliveGreen,opacity=.1] (a1.center)--(a7.center)--(a6.center)--cycle;
  \filldraw[OliveGreen,opacity=.1] (a3.center)--(a2.center)--(a5.center)--cycle;
  \filldraw[OliveGreen,opacity=.1] (a10.center)--(a9.center)--(a6.center)--cycle;
  \filldraw[OliveGreen,opacity=.1] (a11.center)--(a12.center)--(a14.center)--cycle;

  \filldraw[OliveGreen,opacity=.4] (a1.center)--(a2.center)--(a6.center)--cycle;   
  \filldraw[OliveGreen,opacity=.4] (a2.center)--(a5.center)--(a6.center)--cycle;   
  \filldraw[OliveGreen,opacity=.4] (a5.center)--(a6.center)--(a10.center)--cycle;  
  \filldraw[OliveGreen,opacity=.4] (a5.center)--(a10.center)--(a11.center)--cycle; 
  \filldraw[OliveGreen,opacity=.4] (a10.center)--(a11.center)--(a13.center)--cycle;

  \path (a1) -- (a2) coordinate[midway] (m12);
  \path (a2) -- (a6) coordinate[midway] (m26);
  \path (a5) -- (a6) coordinate[midway] (m56);
  \path (a5) -- (a10) coordinate[midway] (m510);
  \path (a10) -- (a11) coordinate[midway] (m1011);
  \path (a11) -- (a13) coordinate[midway] (m1113);

  \draw[dotted, very thick, red] (m12)--(m26)--(m56)--(m510)--(m1011)--(m1113);

  \foreach \p in {m12,m26,m56,m510,m1011,m1113}{
    \fill[red] (\p) circle (2pt);
  }

\end{tikzpicture}
}
    \caption{Example of a 2-Path in a Complex. Here $P $= $(\{1,2\}$,$\{2,6\}$,$\{5,6\}$,$\{5,10\}$,$\{10,11\}$,$\{11,13\})$ is a 2-path from $\{1,2\}$ to $\{11,13\}$.}
    \label{fig:2-path-in-a-complex}
\end{figure}
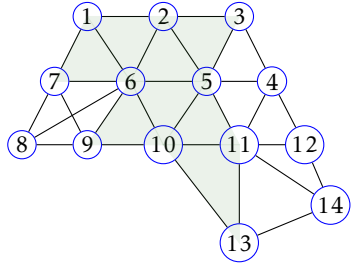

In a weighted $d$-complex \( (X,w) \), the \emph{shortest $d$-path} between two vertices \( \sigma_1 \) and \( \sigma_2 \) is a $d$-path from \( \sigma_1 \) to \( \sigma_2 \) that minimizes the sum of the weights of the $d$-simplices along the $d$-path. The shortest $d$-path may not be unique. The \emph{weighted distance} between two $(d-1)$-simplices \( \sigma_1, \sigma_2 \in X_{(d-1)} \) is the sum of the weights of the $d$-simplices along the shortest $d$-path from \( \sigma_1 \) to \( \sigma_2 \). The weighted distance is denoted as \( d_w(\sigma_1, \sigma_2) \).

If there is no $d$-path between \( \sigma_1 \) and \( \sigma_2 \), the weighted distance is defined to be infinity, i.e., \( d_w(\sigma_1, \sigma_2) = \infty \).

\section{Discovery Algorithms in Simplicial Complexes}
\label{sec:graph-discovery}

Before discussing shortest-path algorithms, we examine two foundational graph-discovery algorithms generalized to simplicial complexes. These procedures, breadth-first search and depth-first search, form the unweighted combinatorial basis of all subsequent weighted algorithms.

\paragraph{Adjacency Graph of $(d-1)$-Simplices}

Let $(X,w)$ be a weighted $d$-complex with a complete $(d-1)$-skeleton.
We define the \emph{adjacency graph} associated with $X$ as
\[
G(X) := (V,E),
\quad
V := X_{d-1}, \qquad
E := \{ (\sigma,\sigma') \in V\times V \mid \sigma \cup \sigma' \in X_d \}.
\]
Each node in $G(X)$ corresponds to a $(d-1)$-simplex, and an edge connects two
nodes if they jointly form a $d$-simplex of $X$.  When $d=1$, $G(X)$ coincides
with the underlying graph itself.

A traversal of $G(X)$ corresponds to the discovery of $(d-1)$-simplices
reachable under adjacency; such traversals preserve the local combinatorial
structure of the complex and establish its connectivity at the $(d-1)$-level.

\subsection{Breadth-First Search (BFS) in a $d$-Complex}

Breadth-first search explores $G(X)$ level by level, discovering all
$(d-1)$-simplices at minimal unweighted $d$-path distance from a given
starting simplex $\sigma_0 \in X_{d-1}$.

\begin{algorithm}[ht]
\caption{Breadth-First Search (BFS) in a $d$-Complex}
\label{alg:bfs-d}
\SetKwInOut{Input}{Input}\SetKwInOut{Output}{Output}
\Input{Connected $d$-complex $X=(X_{d-1},X_d)$, starting simplex $\sigma_0\in X_{d-1}$}
\Output{Traversal order and distance levels of $(d-1)$-simplices from $\sigma_0$}

Mark all $\sigma\in X_{d-1}$ as \emph{unvisited};\\
Initialize an empty queue $Q$;\\
$t(\sigma)\gets +\infty$ for all $\sigma$, \quad $t(\sigma_0)\gets 0$;\\
Enqueue $\sigma_0$ and mark it visited;\\
\While{$Q$ not empty}{
    $\sigma \gets$ Dequeue($Q$);\\
    \ForEach{$\sigma'\in S_\sigma$ with $\sigma\cup\sigma'\in X_d$}{
        \If{$\sigma'$ unvisited}{
            Mark $\sigma'$ as visited\;
            $t(\sigma')\gets t(\sigma)+1$\;
            parent$(\sigma')\gets\sigma$\;
            Enqueue $\sigma'$;
        }
    }
}
\end{algorithm}

\begin{lemma}[Correctness of BFS in $d$-Complexes]
\label{lem:bfs-correctness}
Let $X$ be a connected $d$-complex and $\sigma_0\in X_{d-1}$.  After
Algorithm~\ref{alg:bfs-d} terminates, the value $t(\sigma)$ equals the length
of the shortest unweighted $d$-path from $\sigma_0$ to $\sigma$, for all
$\sigma\in X_{d-1}$.  Moreover, the set of parent relations forms a shortest
$d$-path tree rooted at $\sigma_0$.
\end{lemma}

\begin{proof}
We proceed by induction on the number of layers discovered.
Initially, $t(\sigma_0)=0$.  Suppose the claim holds for all simplices whose
distance from $\sigma_0$ is at most $k$.  When processing such a simplex
$\sigma$, each unvisited neighbor $\sigma'$ with $\sigma\cup\sigma'\in X_d$
is assigned distance $t(\sigma')=t(\sigma)+1=k+1$.  Since every $d$-path
to $\sigma'$ must pass through some $\sigma$ at distance $k$, this value is
minimal.  The FIFO queue ensures that simplices are processed in nondecreasing
order of $t(\cdot)$, hence the invariant is maintained.  Upon termination,
$t(\sigma)$ equals the minimal number of $d$-simplices connecting
$\sigma_0$ and $\sigma$, establishing the result.
\end{proof}

\begin{lemma}[Complexity of BFS]
\label{lem:bfs-complexity}
Let $N:=|X_{d-1}|$ and $M:=|X_d|$.  Algorithm~\ref{alg:bfs-d} runs in time
$\mathcal{O}(N+M)$ and requires space $\mathcal{O}(N+M)$.
\end{lemma}

\begin{proof}
Each $(d-1)$-simplex is enqueued and dequeued at most once.  Each $d$-simplex
is inspected only through its incident $(d-1)$-simplices, implying that
the total number of adjacency checks is linear in $N+M$.
\end{proof}

\subsection{Depth-First Search (DFS) in a $d$-Complex}

Depth-first search performs a recursive traversal of $G(X)$, producing a
spanning forest that captures the hierarchical structure of adjacency among
$(d-1)$-simplices.  This recursive process also serves as the basis for cycle
detection and enumeration.

\begin{algorithm}[ht]
\caption{Depth-First Search (DFS) in a $d$-Complex}
\label{alg:dfs-d}
\SetKwInOut{Input}{Input}\SetKwInOut{Output}{Output}
\Input{Connected $d$-complex $X=(X_{d-1},X_d)$}
\Output{DFS forest and parent relation on $X_{d-1}$}

Mark all $(d-1)$-simplices as unvisited;\\
\ForEach{unvisited $\sigma\in X_{d-1}$}{
    DFS($\sigma$);
}

\BlankLine
\textbf{Procedure} DFS($\sigma$):\\
Mark $\sigma$ as visited;\\
\ForEach{$\sigma'\in S_\sigma$ with $\sigma\cup\sigma'\in X_d$}{
    \If{$\sigma'$ unvisited}{
        parent$(\sigma')\gets\sigma$;
        DFS($\sigma'$);
    }
}
\end{algorithm}

\begin{lemma}[DFS Visits and Forest Structure]
\label{lem:dfs-correctness}
Let $X$ be a connected $d$-complex.
Algorithm~\ref{alg:dfs-d} visits each $(d-1)$-simplex exactly once and
produces a spanning tree of the adjacency graph $G(X)$.
\end{lemma}

\begin{proof}
Marking ensures that each simplex $\sigma$ is visited exactly once, since
recursive calls are made only on unvisited simplices.  Parent pointers define
a tree structure with unique predecessors except the root.  If $X$ has multiple
connected components, the outer loop constructs a spanning forest.
\end{proof}

\begin{lemma}[Complexity of DFS]
\label{lem:dfs-complexity}
Let $N=|X_{d-1}|$ and $M=|X_d|$.  Algorithm~\ref{alg:dfs-d} runs in time
$\mathcal{O}(N+M)$ and space $\mathcal{O}(N+M)$.
\end{lemma}

\begin{proof}
Each adjacency $(\sigma,\sigma')$ is explored at most once.
Recursive depth never exceeds $N$, and the storage for visited markers and
parent pointers is linear in $N+M$.
\end{proof}

\subsection{Essence of Discovery Algorithms}

Both BFS and DFS operate solely on the adjacency relation induced by $X_d$. They generalize classical graph-traversal schemes to the simplicial setting, discovering the structure of $(d-1)$-simplices and their interconnections via $d$-simplices. BFS produces minimal unweighted distances (Lemma~\ref{lem:bfs-correctness}) and provides a natural initialization for weighted algorithms such as Dijkstra and Bellman-Ford, which extend the scalar distance updates to real weights $w(\sigma\cup\sigma')$. DFS, besides its use in connectivity and topological sorting, underlies the cycle-detection and enumeration techniques. Both algorithms have linear complexity in the size of the complex, matching the optimal asymptotic bounds known for graph traversals.

The BFS and DFS algorithms presented above provide the essential discrete
infrastructure for subsequent metric computations on simplicial complexes.
They yield both structural information (connectivity, spanning trees, and
cycles) and the traversal order required by shortest-path procedures,
which are examined in the next section.




\section{Shortest Paths in Absence of Negative Weights}\label{sec:algo}

In this section, we describe and analyze our first algorithm for finding shortest path. We also provide example of the workings of our algorithm. 

Our algorithm is a greedy algorithm to find the shortest $d$-paths from a source $(d-1)$-simplex \( \sigma_0 \) to all other $(d-1)$-simplexes in a weighted complex \( (X, w) \), where \( w: X_d \to \mathbb{R} \) represents the weight function. Note that the weights have to be non-negative for this algorithm to work. The algorithm works by iteratively selecting the $(d-1)$-simplexes with the smallest known distance and updating the distances to their neighbouring $(d-1)$-simplexes. The pseudo-code of our algorithm is given in \cref{alg:dijk-simplicial-complex}. But for ease of understanding, the high-level idea of our algorithm is described in the following steps. 

\begin{flushleft}
   \textbf{Steps of Our Algorithm:} 
\end{flushleft}

\begin{enumerate}
    \item \textbf{Input:}
    A weighted $d$-complex with non-negative weight function $w$ and a starting $(d-1)$-simplex $\sigma_0$. For any $x\in \mathcal{X}_0$, the weight of the $d$-simplex $\sigma_0 \cup \{x\}$ is $w\left(\sigma_0 \cup \{x\}\right)$; Let $w\left(\sigma_0 \cup \{x\}\right) =\infty$ if $\sigma_0 \cup \{x\} \notin X^d$.
    \item \textbf{Initialization:}
    \begin{enumerate}
        \item Set the distance to the source $(d-1)$-simplex \( \sigma_0 \) as \( t(\sigma_0) = 0 \).
        \item Set the distance to all other $(d-1)$-simplexes \( \sigma \in X_{(d-1)} \setminus \{\sigma_0\} \) as \( t(\sigma) = \infty \).
        \item Mark all vertices as unvisited.
    \end{enumerate}
    
    \item \textbf{Iteration:}
    \begin{enumerate}
        \item Select the unvisited $(d-1)$-simplex \( \sigma_1 \) with the smallest known distance \( t(\sigma_1) \). The choice of $\sigma_1$ may be more than one.
        \item For each neighbor \( \sigma_2 \) of \( \sigma_1 \), if the $d$-path from \( \sigma_0 \) to \( \sigma_2 \) through \( \sigma_1 \) offers a shorter $d$-path than the current known distance, update the distance:
        \[
        t(\sigma_2) = \min(t(\sigma_2), t(\sigma_1) + w(\sigma_1, \sigma_2))
        \]
    \end{enumerate}
    
    \item \textbf{Termination:}
    \begin{enumerate}
        \item Mark \( \sigma_2 \) as visited. If all $(d-1)$-simplexes have been visited or the smallest known distance among the unvisited $(d-1)$-simplexes is infinity, the algorithm terminates.
        \item The shortest $d$-path from \( \sigma_0 \) to each $(d-1)$-simplex is now known.
    \end{enumerate}
\end{enumerate}

\begin{algorithm}[!ht] 
	\DontPrintSemicolon
	\KwIn{A weighted $d$-complex with non-negative weight function $w$ and a starting $(d-1)$-simplex $\sigma_{0}$}	
	\KwOut{The shortest $d$-path from $\sigma_{0}$ to each $(d-1)$-simplex}
	\SetKwRepeat{REPEAT}{repeat}{until}
    Set the distance to the source $(d-1)$-simplex \( \sigma_0 \) as \( t(\sigma_0) = 0 \)\\
    Set the distance to all other $(d-1)$-simplexes \( \sigma \in X_{(d-1)} \setminus \{\sigma_0\} \) as \( t(\sigma) = \infty \)\\
    Mark all vertices as unvisited\\

    \While{there exists an unvisited $(d-1)$-simplex $\sigma_{1}$ with $t(\sigma_{1}) < \infty$}{
        Select an unvisited $\sigma_{1}$ minimizing $t(\sigma_{1})$\;
        \ForEach{unvisited neighbor $\sigma_{2}$ of $\sigma_{1}$ (i.e.\ $\sigma_{1}\cup\sigma_{2}\in X^{d}$)}{
            $t(\sigma_{2}) \gets \min\bigl(t(\sigma_{2}),\;t(\sigma_{1}) + w(\sigma_{1}\cup\sigma_{2})\bigr)$\;
        }
        Mark $\sigma_{1}$ as visited\;
    }
    \Return{Distance function $t\colon X_{(d-1)}\to [0,\infty]\,$, where $t(\sigma)$ is the minimum total weight of a chain of $d$-simplices from $\sigma_{0}$ to $\sigma$.}\;

	\caption{\textsc{Shortest $d$-paths in Simplicial Complexes}}
	\label{alg:dijk-simplicial-complex}
\end{algorithm}

For an example of the application of this algorithm, see \cref{fig:algorithm-example}.

\begin{figure}[ht!]
    \centering
    \scalebox{.7}{\begin{tikzpicture}[scale=1.2]

\tikzset{
    bluenode/.style={
        draw=blue,circle, inner sep =1pt
        },
    rednode/.style={
        },
}
  \node[bluenode] (a1) at (-2,0) {1};  
  \node[bluenode] (a2) at (-1,-.6)  {2}; 
  \node[bluenode] (a3) at (-1,.6)  {3};  
  \node[bluenode] (a4) at (0,1.2) {4};  
  \node[bluenode] (a5) at (0,0)  {5};  
  \node[bluenode] (a6) at (0,-1.2)  {6};  
  \node[bluenode] (a7) at (1,.6)  {7};  
  
  \draw (a1) -- (a2); 
  \draw (a1) -- (a3);  
  \draw (a2) -- (a3);  
  \draw (a3) -- (a4);  
  \draw (a3) -- (a5);  
  \draw (a2) -- (a6);
  \draw (a2) -- (a5);
  \draw (a5) -- (a6);  
  \draw (a5) -- (a4);
  \draw (a5) -- (a7);
  \draw (a4) -- (a7);


  \node[bluenode] (a11) at (-2,4) {1};  
  \node[bluenode] (a21) at (-1,3.4)  {2}; 
  \node[bluenode] (a31) at (-1,4.6)  {3};  
  \node[bluenode] (a41) at (0,5.2) {4};  
  \node[bluenode] (a51) at (0,4)  {5};  
  \node[bluenode] (a61) at (0,2.8)  {6};  
  \node[bluenode] (a71) at (1,4.6)  {7};  
  
  \draw (a11) -- (a21); 
  \draw (a11) -- (a31);  
  \draw (a21) -- (a31);  
  \draw (a31) -- (a41);  
  \draw (a31) -- (a51);  
  \draw (a21) -- (a61);
  \draw (a21) -- (a51);
  \draw (a51) -- (a61);  
  \draw (a51) -- (a41);
  \draw (a51) -- (a71);
  \draw (a41) -- (a71);


  \node[bluenode] (a12) at (-2,8) {1};  
  \node[bluenode] (a22) at (-1,7.4)  {2}; 
  \node[bluenode] (a32) at (-1,8.6)  {3};  
  \node[bluenode] (a42) at (0,9.2) {4};  
  \node[bluenode] (a52) at (0,8)  {5};  
  \node[bluenode] (a62) at (0,6.8)  {6};  
  \node[bluenode] (a72) at (1,8.6)  {7};  
  
  \draw (a12) -- (a22); 
  \draw (a12) -- (a32);  
  \draw (a22) -- (a32);  
  \draw (a32) -- (a42);  
  \draw (a32) -- (a52);  
  \draw (a22) -- (a62);
  \draw (a22) -- (a52);
  \draw (a52) -- (a62);  
  \draw (a52) -- (a42);
  \draw (a52) -- (a72);
  \draw (a42) -- (a72);


  \node[bluenode] (a13) at (-2,12) {1};  
  \node[bluenode] (a23) at (-1,11.4)  {2}; 
  \node[bluenode] (a33) at (-1,12.6)  {3};  
  \node[bluenode] (a43) at (0,13.2) {4};  
  \node[bluenode] (a53) at (0,12)  {5};  
  \node[bluenode] (a63) at (0,10.8)  {6};  
  \node[bluenode] (a73) at (1,12.6)  {7};  
  
  \draw (a13) -- (a23); 
  \draw (a13) -- (a33);  
  \draw (a23) -- (a33);  
  \draw (a33) -- (a43);  
  \draw (a33) -- (a53);  
  \draw (a23) -- (a63);
  \draw (a23) -- (a53);
  \draw (a53) -- (a63);  
  \draw (a53) -- (a43);
  \draw (a53) -- (a73);
  \draw (a43) -- (a73);


  \node[bluenode] (a14) at (-2,16) {1};  
  \node[bluenode] (a24) at (-1,15.4)  {2}; 
  \node[bluenode] (a34) at (-1,16.6)  {3};  
  \node[bluenode] (a44) at (0,17.2) {4};  
  \node[bluenode] (a54) at (0,16)  {5};  
  \node[bluenode] (a64) at (0,14.8)  {6};  
  \node[bluenode] (a74) at (1,16.6)  {7};  
  
  \draw (a14) -- (a24); 
  \draw (a14) -- (a34);  
  \draw (a24) -- (a34);  
  \draw (a34) -- (a44);  
  \draw (a34) -- (a54);  
  \draw (a24) -- (a64);
  \draw (a24) -- (a54);
  \draw (a54) -- (a64);  
  \draw (a54) -- (a44);
  \draw (a54) -- (a74);
  \draw (a44) -- (a74);


\node[rednode] at  (-1.4,0) {$2$};
\node[rednode] at  (-.6,0) {$3$};
\node[rednode] at  (.4,.6) {$4$};
\node[rednode] at  (-.4,.6) {$5$};
\node[rednode] at  (-.4,-.6) {$7$};

\node[rednode] at  (-1.4,4) {$2$};
\node[rednode] at  (-.6,4) {$3$};
\node[rednode] at  (.4,4.6) {$\infty$};
\node[rednode] at  (-.4,4.6) {$5$};
\node[rednode] at  (-.4,3.4) {$7$};

\node[rednode] at  (-1.4,8) {$2$};
\node[rednode] at  (-.6,8) {$3$};
\node[rednode] at  (.4,8.6) {$\infty$};
\node[rednode] at  (-.4,8.6) {$\infty$};
\node[rednode] at  (-.4,7.4) {$\infty$};

\node[rednode] at  (-1.4,12) {$2$};
\node[rednode] at  (-.6,12) {$\infty$};
\node[rednode] at  (.4,12.6) {$\infty$};
\node[rednode] at  (-.4,12.6) {$\infty$};
\node[rednode] at  (-.4,11.4) {$\infty$};

\node[rednode] at  (-1.4,16) {$\infty$};
\node[rednode] at  (-.6,16) {$\infty$};
\node[rednode] at  (.4,16.6) {$\infty$};
\node[rednode] at  (-.4,16.6) {$\infty$};
\node[rednode] at  (-.4,15.4) {$\infty$};

\node[rednode] at  (2.4,17) {$w(\{1,2,3\}) = 2,$};
\node[rednode] at  (4.6,17) {$w(\{3,4,5\}) = 5,$};
\node[rednode] at  (6.8,17) {$w(\{2,5,6\}) = 7,$};
\node[rednode] at  (2.4,16.5) {$w(\{2,3,5\}) = 3,$};
\node[rednode] at  (4.6,16.5) {$w(\{4,5,7\}) = 5,$};

\node[rednode] at  (2.8,15.6) {$\sigma_{0} = \{1,3\},\sigma_{0} = \{4,7\}$};
\node[rednode] at  (5.3,15.1) {$S \leftarrow \{\{1,3\}\}$ i.e. except $\{1,3\}$ every 1-simplex is unvisited};

\node[rednode] at  (2.8,12.6) {$S \leftarrow S \cup \{\{1,2\},\{2,3\}\}$};

\node[rednode] at  (3,12.1) {$t(\{1,2\}) = 2$};
\node[rednode] at  (3,11.6) {$t(\{2,3\}) = 2$};

\node[rednode] at  (2.8,8.6) {$S \leftarrow S \cup \{\{2,5\},\{3,5\}\}$};
\node[rednode] at  (3,8.1) {$t(\{2,5\}) = 2+3 = 5$};
\node[rednode] at  (3,7.6) {$t(\{3,5\}) = 2+3 = 5$};

\node[rednode] at  (3.6,4.6) {$S \leftarrow S \cup \{\{3,4\},\{4,5\},\{2,6\},\{5,6\}\}$};
\node[rednode] at  (3.6,4.1) {$t(\{3,4\}) = 2+3+5 = 10$};
\node[rednode] at  (3.6,3.6) {$t(\{4,5\}) = 2+3+5 = 10$};
\node[rednode] at  (3.6,3.1) {$t(\{2,6\}) = 2+3+7 = 12$};
\node[rednode] at  (3.6,2.6) {$t(\{5,6\}) = 2+3+7 = 12$};

\node[rednode] at  (2.8,.6) {$S \leftarrow S \cup \{\{4,7\},\{5,7\}\}$};
\node[rednode] at  (3.6,.1) {$t(\{4,7\}) = 2+3+5+4 = 14$};
\node[rednode] at  (3.6,-.6) {$t(\{5,7\}) = 2+3+5+4 = 14$};

\node[rednode] at  (5.6,-1.6) {$\boxed{d_{w}(\{1,3\},\{4,7\}) = t(\{4,7\}) = 14}$};

\end{tikzpicture}  }
    \caption{Diagram explaining the execution of the algorithm using an example of finding $d_{w}(\{1,3\},\{4,7\})$ in the complex}
    \label{fig:algorithm-example}
\end{figure}

\subsection{Analysis of \cref{alg:dijk-simplicial-complex}}
\paragraph{Correctness} We begin by proving the correctness of our algorithm. For this, we state and prove the following theorem.

\begin{theorem}\label{thm:correct}
Given a weighted $d$-complex with non-negative weight function $w$ and a starting $(d-1)$-simplex $\sigma_0$, \cref{alg:dijk-simplicial-complex} computes \( d_w(\sigma_0, \sigma_1) \) for every \( \sigma_1 \in X_{(d-1)}. \)
\end{theorem}

\begin{proof}
We prove that \cref{alg:dijk-simplicial-complex} correctly computes the shortest $d$-path from the source $(d-1)$-simplex \( \sigma_0 \) to every other $(d-1)$-simplex in a weighted complex \( (X,w) \), where \( w \) is the weight function on the edges. We use Mathematical Induction to show that at each step of Dijkstra's algorithm, the distance \( t(\sigma_1) \) is always the correct shortest $d$-path from the source \( \sigma_0 \) to $(d-1)$-simplex \( \sigma_1 \) for all vertices \( \sigma_1 \in X_{(d-1)} \).

At any step of the algorithm, let $\mathcal{S}$ be the set of $(d-1)$-simplexes whose shortest distance from the source $\sigma_0$ has been finalised. Initially, the set \( \mathcal{S} =\{\sigma_0\} \), and the algorithm sets \( t(\sigma_0) = 0 \) and \( t(\sigma_2) = \infty \) for all \( \sigma_2 \neq \sigma_0 \). Clearly, the distance from \( \sigma_0 \) to \( \sigma_0 \) is correct since \( t(\sigma_0) = 0 \). So, the algorithm correctly initialises the distances.

Assume that the algorithm has correctly computed the shortest $d$-path distances for all vertices in \( \mathcal{S}\), where $|\mathcal{S}| = m$ for some $m\in \mathbb{N}$. Let \( \sigma_1 \) be the next $(d-1)$-simplex added to \( \mathcal{S}\). By the algorithm’s rule, \( \sigma_1 \) is the unvisited $(d-1)$-simplex with the smallest tentative distance \( t(\sigma_1) \). We need to show that this choice of \( \sigma_1 \) guarantees that \( t(\sigma_1) \) is the correct shortest $d$-path from \( \sigma_0 \) to \( \sigma_1 \).

Any $d$-path \( P \) from $\sigma_0$ to $\sigma_1$ must exit \(\mathcal{S}\) before reaching \( \sigma_1 \). The induction hypothesis states that the length of the shortest $d$-path from $S$ to $\sigma_1$ is $t(\sigma_1)$. The induction hypothesis and the choice of $\sigma_1$ also guarantee that a $d$-path visiting any $(d-1)$-simplex outside $\mathcal{S}$ and later reaching $\sigma_1$ has length at least $t(\sigma_1)$. Hence $d_w(\sigma_0,\sigma_1) = t(\sigma_1)$. Therefore, the algorithm correctly computes the shortest $d$-path to \( \sigma_1 \) and adds \( \sigma_1 \) to \( \mathcal{S} \). The algorithm then proceeds to the next $(d-1)$-simplex in the same way, ensuring that at every step, the shortest $d$-paths from \( \sigma_0 \) to all vertices in \( \mathcal{S} \) are correctly computed.

The algorithm terminates when \( \mathcal{S} = X_{d-1}\), at which point \( t(\sigma_1) \) for each $(d-1)$-simplex \( \sigma_1 \) in \( X_{d-1} \) represents the shortest $d$-path from \( \sigma_0 \) to \( \sigma_1 \).

Since the algorithm always selects the $(d-1)$-simplex with the smallest tentative distance at each step and updates distances optimally, the final distances \( t(\sigma_1) \) are indeed the correct shortest $d$-path distances from \( \sigma_0 \) to every $(d-1)$-simplex in \( X_{d-1} \).
\end{proof}

Hence, the correctness of the algorithm \cref{alg:dijk-simplicial-complex} follows from \cref{thm:correct}.

\paragraph{Complexity} Now we deduce the time and memory complexity of \cref{alg:dijk-simplicial-complex}. For clarity we denote the number of $d-1$-simplices by $n = |X_{d-1}|$ and the number of $d$-simplices by $m = |X_{d}|$. We will analyse the algorithm assuming that we use the binary heap data structure to store the distances. Now, if we use a binary heap structure for storing the current distances, then scanning the binary heap to find the smallest distance (\emph{extract-min}) would take $\mathcal{O}(\log n)$ operations per extraction. Next, updating the entry (\emph{decrease-key}) would take $\mathcal{O}(\log n)$ operations per iteration. Hence, the time-complexity of the algorithm is a sum of $n$ extractions and up to $m$ decrease-key operations. Therefore, the overall runtime of the algorithm is given by $\mathcal{O}((n+m)\log n)$.

For each $d-1$-simplex, we store the tentative distance from the starting point, which requires $\mathcal{O}(n)$ memory. Next, for each $d-1$-simplex, we need to store a boolean value that represents whether that simplex has been visited or not. This again requires $\mathcal{O}(n)$ memory. The algorithm also uses a priority queue to quickly extract an unvisited $d-1$-simplex with the smallest distance. The queue stores up to $n$ entries at once, each corresponding to a $d-1$-simplex and its associated distance. Therefore, the queue requires $\mathcal{O}(n)$ memory. Now, to find the neighbours of a given $d-1$-simplex, we must know all the $d$-simplices that contain the given $d-1$-simplex. It requires $\mathcal{O}(m)$ memory to store this adjacency information. This makes the total memory requirement of the algorithm to be $\mathcal{O}(m+n)$.

\section{Generalization of the Bellman-Ford Algorithm}\label{sec:bellman}

We now present a relaxation-based algorithm that generalizes the classical
Bellman-Ford procedure to the higher-dimensional setup introduced in this
paper. In contrast to Algorithm~\ref{alg:dijk-simplicial-complex}, which assumes
non-negative $d$-simplex weights, the following method correctly computes
shortest weighted $(d-1)$-paths even in the presence of negative weights and
detects the existence of negative $d$-cycles.

\medskip
Let $(X,w)$ be a weighted $d$-complex with the vertex set
$V := X_{d-1}$ and the adjacency relation defined as before:
two $(d-1)$-simplices $\sigma,\sigma'\in V$ are \emph{neighbors} if and only
if $\sigma\cup\sigma'\in X_d$. The cost associated with the pair
$(\sigma,\sigma')$ is the weight $w(\sigma\cup\sigma')$ of that $d$-simplex.
Define the edge set
\[
E := \{ (\sigma,\sigma') \in V\times V \mid \sigma\cup\sigma' \in X_d \},
\qquad c(\sigma,\sigma') := w(\sigma\cup\sigma').
\]
A walk of $(d-1)$-simplices
$\Pi=(\sigma_0,\sigma_1,\dots,\sigma_\ell)$ with
$(\sigma_{i-1},\sigma_i)\in E$ for all $i$ is called a \emph{$d$-path},
and its total weight is
\(
w(\Pi) := \sum_{i=1}^{\ell} c(\sigma_{i-1},\sigma_i).
\)
A \emph{negative $d$-cycle} is a closed $d$-path $\Pi$ with $w(\Pi)<0$.

\begin{algorithm}[ht]
\caption{Bellman-Ford on Weighted $d$-Complexes (BF-$d$)}
\label{alg:bf-d}
\SetKwInOut{Input}{Input}\SetKwInOut{Output}{Output}
\Input{Weighted $d$-complex $(X,w)$, source $\sigma_0 \in X_{d-1}$}
\Output{Shortest distances $t(\sigma)$ from $\sigma_0$; or report a negative $d$-cycle}

$V \gets X_{d-1}$; \quad
$E \gets \{(\sigma,\sigma') : \sigma\cup\sigma' \in X_d\}$\;

\ForEach{$\sigma\in V$}{ $t(\sigma)\gets +\infty$ }
$t(\sigma_0)\gets 0$\;

\BlankLine
\For{$i\gets 1$ \KwTo $|V|-1$}{
    \ForEach{$(\sigma,\sigma')\in E$}{
        \If{$t(\sigma)\neq +\infty$ \textbf{and} $t(\sigma') > t(\sigma)+c(\sigma,\sigma')$}{
            $t(\sigma') \gets t(\sigma)+c(\sigma,\sigma')$\;
        }
    }
}

\BlankLine
$U \gets \varnothing$\;
\ForEach{$(\sigma,\sigma')\in E$}{
    \If{$t(\sigma)\neq +\infty$ \textbf{and} $t(\sigma') > t(\sigma)+c(\sigma,\sigma')$}{
        $U \gets U \cup \{\sigma'\}$\;
    }
}
\BlankLine
\If{$U\neq\varnothing$}{
    \textbf{Report:} ``Negative $d$-cycle detected'' and return the set of all
    $(d-1)$-simplices reachable from $U$ under $E$\;
}
\Else{
    \Return{$t(\cdot)$}\;
}
\end{algorithm}

\paragraph{Discussion}
Algorithm~\ref{alg:bf-d} preserves the same combinatorial interpretation as
Algorithm~\ref{alg:dijk-simplicial-complex}: relaxation is performed over the adjacency
induced by $X_d$, but instead of greedily fixing vertices, all edges are
relaxed iteratively. The procedure tolerates negative simplex weights and
identifies negative $d$-cycles via the standard one-pass verification step.

\begin{example}[Negative $d$-Cycle Detection]
Consider a $2$-complex $X$ whose $1$-simplices (edges)
$\{a,b,c\}$ bound two $2$-simplices $\tau_1,\tau_2$
sharing a common edge $b$. Assign weights
$w(\tau_1)=-1$, $w(\tau_2)=-2$.
The adjacency relation on $V=X_1=\{a,b,c\}$ contains the edges
$(a,b)$, $(b,c)$, $(c,a)$ with total cycle weight
$w(\tau_1)+w(\tau_2)=-3<0$.
Executing Algorithm~\ref{alg:bf-d} with source $a$ yields
decreasing distances after $|V|-1$ relaxations, and
the final verification pass identifies $a,b,c\in U$,
signalling a negative $2$-cycle corresponding to the two faces $\tau_1,\tau_2$.
\end{example}

\begin{lemma}[Correctness of BF-$d$]
\label{lem:bf-correctness}
Let $(X,w)$ be a finite weighted $d$-complex and let
$\sigma_0\in X_{d-1}$. If $(X,w)$ contains no negative $d$-cycle,
then upon termination Algorithm~\ref{alg:bf-d} satisfies
\[
t(\sigma)
= \min_{\Pi:\,\sigma_0\leadsto\sigma} w(\Pi)
\quad\text{for all } \sigma\in X_{d-1},
\]
where the minimum is taken over all $d$-paths from $\sigma_0$ to $\sigma$.
If $(X,w)$ contains a negative $d$-cycle,
then $U\neq\varnothing$ after the final iteration, and every element of $U$
lies on or is reachable from some negative $d$-cycle.
\end{lemma}

\begin{proof}
The argument is identical to the classical Bellman-Ford invariant proof,
with the graph of $(d-1)$-simplices replacing the vertex set.
Let $\Pi_k(\sigma)$ denote the family of $d$-paths from $\sigma_0$ to $\sigma$
containing at most $k$ edges.  We prove by induction on $k$
that after the $k$-th relaxation round,
$t(\sigma)\le \min_{\Pi\in\Pi_k(\sigma)}w(\Pi)$ for all $\sigma$.
For $k=0$ this is immediate. For the inductive step, relaxing each
$(\sigma',\sigma)\in E$ allows one additional edge,
so $t(\sigma)\le t(\sigma')+c(\sigma',\sigma)
\le \min_{\Pi\in\Pi_{k-1}(\sigma')}w(\Pi)+c(\sigma',\sigma)
=\min_{\Pi\in\Pi_k(\sigma)}w(\Pi)$. After $|V|-1$ rounds,
every simple path (of length at most $|V|-1$) has been considered.
If no negative cycle exists, a shortest simple path attains the optimum.
If a negative cycle exists, an additional relaxation decreases
$t(\sigma)$ for some $(\sigma,\sigma')$, placing $\sigma'\in U$.
Thus the lemma follows.
\end{proof}

\paragraph{Complexity}
Let $n:=|X_{d-1}|$ and $m_d:=|X_d|$.
Each $d$-simplex induces at most $(d+1)d$ neighbor pairs.
Denoting $m_e:=|E|=\Theta(d^2 m_d)$,
Algorithm~\ref{alg:bf-d} performs $\mathcal{O}(n)$ passes over $m_e$ edges,
giving a total time complexity $\mathcal{O}(n\cdot m_e)=\mathcal{O}(n\,d^2 m_d)$ and
space complexity $\mathcal{O}(n+m_d)$.

\section{Real-life Example of Our Setup}\label{sec:realexample}

In this section we provide a real-life example of our setup that shows how our setup can better model some of the real-life problems than graphs.

In modern communication systems, higher-order interactions often arise naturally, particularly in multi-agent or wireless sensor networks, where groups of devices collaborate to sense, process, or transmit data. These scenarios go beyond the pairwise connectivity that ordinary graphs can express, and are therefore more accurately represented by simplicial complexes.

Consider a distributed network of sensor nodes or autonomous drones deployed over a geographical region. Each sensor node represents a $0$-simplex. Two nodes that can communicate directly via a wireless link form a $1$-simplex. Moreover, subsets of three or more nodes that are capable of joint sensing or cooperative transmission form higher-dimensional simplices. Thus, a 2-simplex $\{i,j,k\}$ indicates a cooperative triplet of sensors that can jointly process or relay data, and more generally, a $d$-simplex represents a $d$-agent coalition capable of group-level communication or computation.

In this setting, a weighted $d$-complex $(X,w)$ arises naturally, where each $d$-simplex $\tau \in X_d$ carries a weight $w(\tau)$ representing the total energy cost, delay, or bandwidth consumption associated with that cooperative operation. Two $(d-1)$-simplices $\sigma,\sigma'$ are said to be adjacent if $\sigma \cup \sigma' \in X_d$, i.e., if they share participation in some $d$-agent transmission. The communication process from one coalition of agents to another can therefore be viewed as a $d$-path in the complex.

The algorithms introduced in this paper can then be interpreted in this framework. The generalized BFS and DFS algorithms on $d$-complexes discover all clusters of agents reachable under cooperative adjacency, thereby serving as topological discovery tools for multi-agent connectivity. The simplicial Dijkstra algorithm (see \cref{alg:dijk-simplicial-complex}) computes the minimum-cost cooperative route for data transmission from one agent group to another, optimizing total communication energy or latency across multiple overlapping coalitions. The generalized Bellman–Ford procedure (see \cref{alg:bf-d}) extends this to time-varying or uncertain environments where communication costs may fluctuate or even become negative (e.g., due to energy harvesting or incentive mechanisms), while also detecting unstable feedback loops corresponding to negative $d$-cycles.

In summary, the shortest $d$-path problem in weighted simplicial complexes provides a rigorous combinatorial model for cooperative routing in multi-agent systems. By extending classical shortest-path computations to higher-dimensional adjacency, this framework captures collective communication patterns that cannot be modeled by graphs alone. Consequently, the proposed algorithms have potential applications in wireless sensor networks, distributed robotics, and cooperative multi-agent systems where higher-order interactions play a central role.

\section{Conclusion}\label{sec:conclu}
In this study, we have defined the concept of weight in a $d$-complex to facilitate our study of $d$-complex as an object in discrete mathematics and graph theory rather than a topological object. The concept of a $d$-path comes from this naturally. The main contributions of our paper are our novel algorithms to find the shortest path in a weighted $d$-complex. With its vast application in various fields, our algorithms are poised to play a crucial role in the study of complexes as mathematical objects with real-life applications. Our algorithms can be seen as the $d$-complex version of Dijkstra's Algorithm and Bellman-Ford Algorithm for graphs. Analogy of a weighted $d$-complex with a weighted graph opens numerous questions in this field of study. The concepts introduced in this paper pave the path for a line of study in which simplicial complexes will be viewed as geometric computational objects rather than abstract entities.

This work extends beyond just providing a thorough analysis of our algorithms, this also indicates possible integration of simplicial complexes with various concepts of graph theory, specially weighted graphs. 

One future aspect of study is the concept of a spanning $d$-tree, where similar to graphs, the case of absence of a cycle in a connected simplicial complex can be defined to be a $d$-tree. The natural question to ask is, given $d$-complex how can we decide if the given $d$-complex is $d$-tree or not. Furthermore, if a $d$-complex is given, how does one identify a minimal spanning $d$-tree? The mono-dimensional version of these questions are of immense interest in the study of theoretical computer science and graph theory, which make these questions raised above worth pursuing.

\section{Declarations}

\subsection*{Funding}
No funding was received for conducting this study.

\subsection*{Conflict of Interest}
The authors declare that they have no potential conflict of interest.

\subsection*{Data Availability}
No datasets were generated or analyzed during the current study.

\bibliography{references}

\end{document}